\documentclass[letterpaper, 10 pt, conference]{ieeeconf}  

\IEEEoverridecommandlockouts                              

\usepackage{cite}
\usepackage{amsmath,amssymb,amsfonts}
\usepackage{algorithmic}
\usepackage{graphicx}
\usepackage{textcomp}
\usepackage{xcolor}
\usepackage{caption}
\usepackage{subcaption}
\newtheorem{mydef}{Definition}
\newtheorem{prb}{Problem}

\newtheorem{thm}{Theorem}
\newtheorem{lem}{Lemma}

\newtheorem{conj}{Conjecture}
\usepackage{optidef}
\usepackage{tikz}
\usetikzlibrary{patterns,angles,quotes,arrows}
\usepackage{pgfplots}
\usepackage[ruled,vlined]{algorithm2e}
\usepackage{geometry}
\usepackage{accents}

\usetikzlibrary{shapes.misc}
\tikzset{cross/.style={cross out, draw=black, minimum size=2*(#1-\pgflinewidth), inner sep=0pt, outer sep=0pt},
cross/.default={1pt}}
\def\BibTeX{{\rm B\kern-.05em{\sc i\kern-.025em b}\kern-.08em
    T\kern-.1667em\lower.7ex\hbox{E}\kern-.125emX}}

\title{\LARGE \bf
Identifying Single-Input Linear System Dynamics from Reachable Sets
}

\author{{Taha Shafa}
\and{Roy Dong}
\and{Melkior Ornik}
\thanks{TS and MO are with the Department of Aerospace Engineering at the University of Illinois Urbana-Champaign. RD is with the Department of Electrical and Computer Engineering at the University of Illinois Urbana-Champaign.
        {\tt\small \{tahaas2, roydong, mornik\}@illinois.edu}}
}

\begin{document}

\maketitle
\thispagestyle{empty}
\pagestyle{empty}

\begin{abstract}

This paper is concerned with identifying linear system dynamics without the knowledge of individual system trajectories, but from the knowledge of the system's reachable sets observed at different times. Motivated by a scenario where the reachable sets are known from partially transparent manufacturer specifications or observations of the collective behavior of adversarial agents, we aim to utilize such sets to determine the unknown system's dynamics. This paper has two contributions. Firstly, we show that the sequence of the system's reachable sets can be used to uniquely determine the system's dynamics for asymmetric input sets under some generic assumptions, regardless of the system's dimensions. We also prove the same property holds up to a sign change for two-dimensional systems where the input set is symmetric around zero. Secondly, we present an algorithm to determine these dynamics. We apply and verify the developed theory and algorithms on an unknown band-pass filter circuit solely provided the unknown system's reachable sets over a finite observation period.

\end{abstract}

\section{Introduction}
This paper aims to determine whether it is possible to use a control system's reachable sets obtained at different time instances to calculate the system's dynamics. In certain instances, we may be able to determine an approximation of a system's reachable sets over a finite observation period. The purpose of this paper is to show that such information can be utilized to arrive at a dynamic model for an unknown system. Practical applications may include system identification of high-density drone and missile swarms \cite{asaamoning2021drone, taub2013intercept} where the reachable set can be found by observing multiple agents collectively, but without the capability of distinguishing them. Other applications include predicting macro-level population behaviors, e.g., determining how crowd behavior changes under social or economic events like the introduction of a new population or changes in the stock market \cite{musse1997model}. We may also be able to model internal body functions on the cellular level \cite{schiebinger2017reconstruction, zhang2021optimal}, namely understanding how cells change their identity and behavior in living systems.

We must first show that model identification using reachable sets will uniquely determine an unknown system's true dynamics. After uniqueness is proven, we develop a method to identify a linear model of an unknown system's behavior using its reachable sets. Previous research in system identification presents the most closely related contributions to the method presented in this paper. However, previous work on system identification classically relies on frequency response techniques induced by randomized actuator inputs \cite{ljung1998system, pintelon2012system}. More sophisticated system identification techniques involve neural networks \cite{miller1995neural}. Single-layer and multi-layer neural networks have also been applied with the use of parameter estimation algorithms using a single hidden layer \cite{chen1990non} and $H_\infty$ control-induced excitations for robust identification of system nonlinearities \cite{480609}. More recent work involves using recurrent neural networks \cite{medsker2001recurrent, zaremba2014recurrent} with Long Short-Term Memory Units (LSTM) and fractional order neural networks (FONN) \cite{aguilar2020fractional, zuniga2021fractal} to identify and control dynamic systems. These methods, however, cannot be used unless one has access to a system's actuators or individual trajectories. The significant difference of our novel method is that it does not require control of any actuators to model an unknown system nor observations of individual trajectories.

On a high level, the problem in this paper involves identifying the behaviors or capabilities of an observed system under limited information. While there exist other methods for adversarial behavior recognition, those works are focused on determining adversarial agent goals by matching actions of an agent against a plan library \cite{ledezma2004predicting,lisy2012game,le2015generative}.  More recent work \cite{iglesias2010evolving,zhang2021anticoncealer} proposes using evolving fuzzy systems and artificial intelligence to adaptively predict agent behavior. In contrast, our method is starkly different since it is not primarily concerned with predicting adversarial behavior, but determining all possible actions of an adversary within a time horizon. Thus, instead of using a library of finite predetermined adversarial actions, our method uses reachable sets to produce a dynamic model of an unknown system.

The outline of this paper is as follows: in Section II, we discuss the problem statement, namely posing the question of whether linear dynamics can be uniquely recovered given an unknown system's sequence of reachable sets and how to recover said dynamics. In Section III, we address the question of whether the system dynamics are uniquely determined by the system's reachable sets. We show that under generic assumptions, the system dynamics are indeed unique under asymmetric input sets. For unknown systems with input sets symmetric around zero, uniqueness modulo a sign has been proved in the two-dimensional case; we conjecture the same holds for higher dimensions. In Section IV, we propose a procedure using knowledge of the reachable sets to calculate the system dynamics. In Section V, we illustrate by example how to implement this procedure to identify the models of an unknown band-pass filter circuit and an additional dynamical system with a symmetric input set.

\subsection{Notation}
We denote the set of all $n\times m$ real and complex matrices by $\mathbb{R}^{n\times m}$ and $\mathbb{C}^{n\times m}$ respectively; for $M \in \mathbb{R}^{n\times m}$, we let $M^T \in \mathbb{R}^{m\times n}$ denote its transpose. Vectors $e_1,\hdots,e_n$ will denote the canonical basis vectors in $\mathbb{R}^n$. We let $\mathbb{N}$ denote the set of all natural numbers, $\mathbb{Z}_{\geq 0}$ denote the set of non-negative integers, and $GL(n)$ denote the set of invertible square matrices of dimension $n\in\mathbb{N}$. Let $\mathcal{S}$ be a set of points in $\mathbb{R}^n$. Then $\mathrm{Conv}(\mathcal{S})$ denotes the convex hull of $\mathcal{S}$. Notation $B\mathcal{X}$ where $B \in \mathbb{R}^{n\times m}$ and $\mathcal{X} \subset \mathbb{R}^m$ denotes the set $B\mathcal{X} = \{Bx~|~x\in \mathcal{X}\}$. Given two sets $\mathcal{A},\,\mathcal{B}\in\mathbb{R}^n$, we denote $\mathcal{A} \oplus \mathcal{B} = \{a + b~|~a\in\mathcal{A},\,b\in\mathcal{B}\}$ as their Minkowski sum. Similarly, $\mathcal{A} \ominus \mathcal{B} = \{c \in \mathbb{R}^n~|~c\oplus\mathcal{B}\subseteq\mathcal{A}\}$ denotes the Minkowski difference. We also define $\mathcal{A} + b = \{a + b~|~a\in\mathbb{R}^n\}$ as the translation of $\mathcal{A}$ by $b \in \mathbb{R}^n$.
\section{Problem Statement}
We consider the discrete-time, single-input linear system 

\vskip -5pt

\begin{equation}\label{discreteLinearSystem}
    x[i+1] = Ax[i] + bu[i], \quad x[0] = 0,
\end{equation}

\noindent where all $i \in \mathbb{Z}_{\geq 0}$, $x \in \mathbb{R}^n$, $A \in \mathbb{R}^{n\times n}$, $b\in \mathbb{R}^n$ and $u~\in~\mathcal{U}~\subset~\mathbb{R}$ where $\mathcal{U} = [\underline{u}, \overline{u}]$ such that $\underline{u} \neq \overline{u}$. We assume $b\neq 0$ since the system's reachable sets are trivial otherwise. We also assume $x[0] = 0$; by a shift in coordinates, the case of $x[0] \neq 0$ is equivalent to that of an \textit{affine} system $x[i+1]=Ax[i]+bu[i]+c$ with initial state at the origin. Solving the problem in this setting can likely be approached by reproducing similar calculations in subsequent sections, but we leave such an effort for future work.

Our goal is to establish whether the dynamics of \eqref{discreteLinearSystem}, i.e., matrices $A$ and $b$, can be determined using the system's reachable sets. We now formally define said reachable sets. 

\begin{mydef}\label{reachableSet_Definition}
For $i \in \mathbb{Z}_{\geq 0}$, the (forward) reachable set of system \eqref{discreteLinearSystem} at time $i$ is $$\mathcal{R}(i, x[0]) = \{\phi_u(i; x[0]) ~ | ~ u : \mathbb{Z}_{\geq 0} \to \mathcal{U}\},$$ where $\phi_u(\cdot; x[0])$ denotes the controlled trajectory of system \eqref{discreteLinearSystem} with control signal $u$.
\end{mydef}

We present the problem of whether the system dynamics are uniquely determined by the system's reachable sets.

\begin{prb}\label{problemStatement}
Given a sequence of sets $\{\mathcal{R}(i,0)\}_{i\in\mathbb{N}}$ which is generated by \eqref{discreteLinearSystem} for some $(A,b)$, determine whether $(A,b)$ can be uniquely recovered from $\{\mathcal{R}(i,0)\}_{i\in\mathbb{N}}$.
\end{prb}

Notice that we explicitly assume the knowledge of all reachable sets at all times. Such an assumption might not always be realistic. We will show that we often need only the first $n + 1$ reachable sets to uniquely recover the dynamics. We leave the more general case -- where only reachable sets at different time steps are available -- for future work.

The first step to solving Problem \ref{problemStatement} is to derive a simple relationship between the system matrices and $\mathcal{R}(i,0)$. Given system \eqref{discreteLinearSystem}, we naturally utilize Minkowski sums and the Minkowski difference \cite{althoff2015computing} to produce such a relationship for all $i \in \mathbb{N}$. 

\begin{thm}\label{discreteMinkowskiDifference}
Let $\mathcal{R}(i,0)$ be the reachable set at time $i$ of \eqref{discreteLinearSystem}. Then 

\begin{equation}\label{theoremABU}
    A^{i-1}b\mathcal{U} = \mathcal{R}(i,0) \ominus \mathcal{R}(i-1,0).
\end{equation}
\end{thm}

\begin{proof}
By \eqref{discreteLinearSystem} it is clear that $\mathcal{R}(1,0) = b\mathcal{U}$. Since $$x[i] = A^ix[0] + A^{i-1}bu[0] + \hdots + bu[i-1],$$ clearly $$\mathcal{R}(i,0) = A^{i-1}b\mathcal{U} \oplus \hdots \oplus b\mathcal{U}$$ and hence $$\mathcal{R}(i,0) = A^{i-1}b\mathcal{U} \oplus \mathcal{R}(i-1,0).$$

We recall that the Minkowski sum of two convex sets is also convex \cite{fradelizi2022volume}. Since all sets $A^{i-1}b\mathcal{U}$ are convex by the definition of $\mathcal{U}$, all sets $\mathcal{R}(i,0)$ are convex by induction. Hence, the appropriate Minkowski difference \cite{yan2015closed} can be calculated to arrive at \eqref{theoremABU}. 
\end{proof}

Theorem \ref{discreteMinkowskiDifference} implies that we can obtain $\{A^{i-1}b\mathcal{U}\}_{i\in\mathbb{N}}$ using the reachable sets $\mathcal{R}(i,0)$. We will prove that when $\mathcal{U} \neq [-c,c]$ with $c \in \mathbb{R}$, matrices $A$ and $b$ are indeed \textit{generically} uniquely defined from $\{A^{i-1}b\mathcal{U}\}_{i\in\mathbb{N}}$, that is, uniquely defined under the assumptions formally written in Theorem \ref{Theorem_generalUniqueMatrix} shown to be generic in a topological sense in Lemma \ref{matrixProperties_GeneralCase}. When $\mathcal{U} = [-c,c]$ for some $c \in \mathbb{R}$, we can show that $(A,b)$ are not uniquely defined, but conjecture that they are unique up to a change in sign. We prove that this property holds for $n = 2$. We shall refer to solutions for cases with such a set $\mathcal{U}$ as \textit{$\pm$-unique}, which is explicitly defined in the next section. 

Following Problem \ref{problemStatement}, which seeks to determine whether system dynamics are uniquely defined from reachable sets, we present the second problem, which aims to explicitly determine such dynamics.

\begin{prb}\label{ProblemStatement2}
Develop a method to recover at least one pair $(A,b)$ which generates $\{\mathcal{R}(i,0)\}_{i\in\mathbb{N}}$.
\end{prb}

Based on methods in \cite{althoff2015computing} for calculating Minkowski differences, we can calculate $\{A^{i-1}b\mathcal{U}\}_{i \in \mathbb{N}}$. We show in Section IV that the results of these Minkowski differences and knowledge of $\mathcal{U}$ are sufficient for calculating $\{A^{i-1}b\}_{i \in \mathbb{N}}$, which in turn can be utilized to calculate the matrix pair $(A,b)$ for controllable systems. We first tackle Problem \ref{problemStatement}.
\section{Uniqueness of the Derived System Model}

We wish to determine when any pair $(A,b)$ uniquely defines the dynamics of \eqref{discreteLinearSystem}. It can be easily shown that the answer is generally negative. Consider an unknown system \eqref{discreteLinearSystem} where

\vskip -5pt

\begin{equation}\label{example}
    A = \begin{bmatrix}0 & 0\\ 0 & 1\end{bmatrix},\quad b = \begin{bmatrix}0 \\ 1\end{bmatrix}
\end{equation}

\noindent and $\mathcal{U} = [0,1]$. By equation \eqref{theoremABU} of Theorem \ref{discreteMinkowskiDifference}, we see that if $A' = I$, then the reachable sets of \eqref{discreteLinearSystem} with matrix pairs $(A,b)$ and $(A',b)$ are equivalent. Thus, we begin by determining sufficient conditions which guarantee whether $(A,b)$ can be uniquely recovered as stated in Problem \ref{problemStatement}. We will show uniqueness under several technical assumptions; Lemma \ref{matrixProperties_GeneralCase} shows said assumptions are generic in a topological sense.

\begin{lem}\label{matrixProperties_GeneralCase}
Let $\mathcal{N} \subset \mathbb{R}^{n\times n}$ be the set of all matrices such that if $A \in \mathcal{N}$, then $A^2$ has distinct eigenvalues. Let $b\in\mathbb{R}^n\backslash\{0\}$ and $\mathcal{O} \in \mathbb{R}^{n\times n}$ be the set of all matrices such that, if $A \in \mathcal{O}$ and $\eta \in \mathbb{C}^n$ is any left eigenvector of $A$, $b^T\eta \neq 0$. Then, $GL(n)\cap\mathcal{N}\cap\mathcal{O}$ is an open and dense set. 
\end{lem}
\begin{proof}
It is a well known result that the set of all matrices with distinct eigenvalues and the set $GL(n)$ are both open and dense \cite{kelley2017general}. Clearly, openness of the former set implies $\mathcal{N}$ is open. To show $\mathcal{N}$ is also dense, we would follow similar steps as part of the proof to show $GL(n)\cap\mathcal{N}\cap\mathcal{O}$ is dense. For succinctness, we prove $GL(n)\cap\mathcal{N}\cap\mathcal{O}$ is open and dense and leave the proof that $\mathcal{N}$ is dense to the reader.

Openness of $GL(n)\cap\mathcal{N}\cap\mathcal{O}$ can be trivially concluded by the continuity of eigenvectors \cite{kato1976perturbation}, meaning if we consider a matrix $A(t)$ whose elements are a continuous function of $t$, any eigenvectors $v_i(t)$ and left eigenvectors $\eta_i(t)$ of norm $1$ of $A(t)$ are continuous function of $t$.

We now prove denseness. In other words, we will show that for any arbitrary matrix $A$ and any $\epsilon > 0$, there exists a matrix $A'' \in GL(n)\cap\mathcal{N} \cap \mathcal{O}$ such that $\|A~-~A''\|~<~\epsilon$. Let $\eta_i$ be the left eigenvectors of $A$ so that $\eta_i^TA~=~\lambda_i\eta_i^T$. By the denseness of $GL(n)$, for any $\delta > 0$ we can find vectors $\eta_i'^{ T}$ such that $\|\eta_i^T~-~\eta_i'^{T}\|~<~\delta$ for all $i$ and $\mathrm{det}([\eta_1'\textrm{ } \eta_2'\textrm{ } \cdots\textrm{ } \eta_n'])~\neq~0$. By the continuity of determinants and because $b\neq 0$, we can slightly perturb one element of $\eta_i^{\prime T}$ to obtain $\eta_i''^T$ such that $\mathrm{det}([\eta_1''\textrm{ } \eta_2''\textrm{ } \cdots\textrm{ } \eta_n''])~\neq~0,\,\|\eta_i^T - \eta_i''^T\| < \delta$ for all $i$, and $b^T\eta_i\neq 0$. We now let $\eta_i''$ form a basis in $\mathbb{C}^n$, and define a matrix $A^\prime$ such that $\eta_i''^TA^\prime = \lambda_i\eta_i''^T$ and $A^\prime \in \mathcal{O}$. If the perturbations above are performed in a way that ensure that perturbations of real eigenvectors remain real, and perturbations of complex conjugate vectors remain complex conjugates, matrix $A'$ is real \cite{strang2006linear}.

Since $\eta_i''$ form a basis in $\mathbb{C}^n$, we can represent any vector $x \in \mathbb{R}^n$ as $x = \sum_{i=1}^n\beta_i(x)\eta_i''^T$ where $\beta_i(x) \in \mathbb{R}$. We can compute $\beta_i(x)$ as a continuous function of $x$. Recall that $\|A\| = \max_{\|x\| = 1}\|Ax\|$. We consider $x$ such that $\|x\| = 1$. Then $\beta_i(x)$ is a continuous function on a compact space and thus has a maximum. Let $\alpha_i = \max\{|\beta_i(x)|~|~\|x\| = 1\}$. Note that $x^TA^\prime = \sum_{i=1}^n\lambda_i\beta_i(x)\eta_i''^T$. It follows that $$\|x^TA - x^TA^\prime\| \leq \sum_{i=1}^n\|(\beta_i(x)\eta_i''^T)A - \beta_i(x)\lambda_i\eta_i''^T\|$$  $$= \sum_{i=1}^n\|\beta_i(x)(\eta_i^TA - (\eta_i^T-\eta_i''^T)A) - \beta_i(x)\lambda_i\eta_i''^T\|$$ $$= \sum_{i=1}^n\|\beta_i(x)((\eta_i''^T - \eta_i^T)A + \lambda_i(\eta_i^T - \eta_i''^T))\| $$ $$< \sum_{i=1}^n(\|\alpha_i A\| + \|\alpha_i\lambda_i\|)\delta$$ and so if we set $\delta = \epsilon/(2\sum_{i=1}^n\|\alpha_i A\| + \|\alpha_i\lambda_i\|)$, then $\|x^TA - x^TA^\prime\| < \epsilon/2$. 

Given $\lambda_1,\hdots,\lambda_n$, for any $\rho > 0$ we can obviously find a set $\{\lambda_1',\hdots,\lambda_n'\}$ such that $|\lambda_i - \lambda_i'| < \rho$ for all $i$, $\lambda_i' \neq 0$ for all $i$, and $\lambda_i = \overline{\lambda_j}$ implies $\lambda_i' = \overline{\lambda_j'}$. Now, define $A''$ by  $\eta_i''^TA'' = \lambda_i'\eta_i''^T$ and $A'' \in \mathcal{M} \cap \mathcal{N} \cap \mathcal{O}$. As before, if the perturbation of eigenvalues is performed in such a way that real eigenvalues remain real and complex conjugates remain conjugate, $A''$ is real. It follows that $$\|x^TA' - x^TA''\| \leq \sum_{i=1}^n\|\beta_i(x)\lambda_i\eta_i''^T - \beta_i(x)\lambda_i'\eta_i''^T\|$$ $$= \sum_{i=1}^n\|\beta_i(x)\eta_i''^T(\lambda_i - \lambda_i')\| < \sum_{i=1}^n\|\alpha_i\eta_i''^T\|\rho.$$ If we set $\rho = \epsilon/(2\sum_{i=1}^n\|\alpha_i\eta_i''^T\|)$, then $\|x^TA' - x^TA''\| < \epsilon/2$. Finally we have $\|x^TAx - x^TA''\| = \|x^TA - x^TA'' + x^TA' - x^TA'\| \leq \|x^TA - x^TA'\| + \|x^TA' - x^TA''\| < \epsilon/2 + \epsilon/2 = \epsilon$. Since this inequality holds for all $x$ such that $\|x\| = 1$, indeed $\|A - A''\| < \epsilon$, and the claim is proven.
\end{proof}

We emphasize that many well-known linear controllable systems, such as the discrete double integrator, RLC circuit, and linearized pendulum \cite{amato2005finite}, contain $A$ matrices which satisfy the conditions of Lemma~\ref{matrixProperties_GeneralCase}. Also, these generic assumptions are not necessary, but sufficient to guarantee uniqueness. For example, a row perturbation of $A$ in \eqref{example} clearly does not satisfy the generic assumptions in \mbox{Lemma \ref{matrixProperties_GeneralCase}}, but the reachable sets of \eqref{discreteLinearSystem} with this new matrix can be used to uniquely generate the dynamics, which implies this method can be applied to a larger set of systems. Finding such non-generic assumptions which guarantee uniqueness is a highly involved problem and remains for future work. In the proof below, we will use the assumptions in \mbox{Lemma \ref{matrixProperties_GeneralCase}} to prove that the dynamics derived from reachable sets are generically unique, at least for an asymmetric input set.

\begin{thm}\label{Theorem_generalUniqueMatrix}
Let $\mathcal{U} = [c,d]$, where $c \neq \pm d$. Let $\eta_i~\in~\mathbb{C}^n$ for $i\in\{1,\ldots,n\}$ be the left eigenvectors of $A$. Let the sequence $\{\mathcal{R}(j,0)\}_{j\in\mathbb{N}}$ be generated by system \eqref{discreteLinearSystem} for system matrices $(A,b)$ and $(A',b')$, where $A,\,A' \in GL(n)$, $A$ and $A'$ have $n$ distinct eigenvalues, and $b^T\eta_i \neq 0$ for all $i$. Then, $(A,b) = (A',b')$.
\end{thm}
\begin{proof}
If $(A,b)$ and $(A^\prime,b^\prime)$ for system \eqref{discreteLinearSystem} produce an identical sequence $\{\mathcal{R}(j,0)\}_{j\in\mathbb{N}}$, then $\mathcal{R}(1,0) = b\mathcal{U} = b^\prime\mathcal{U}$, i.e., there are two options: (i) $bc = b'c$ and $bd = b'd$ or (ii) $bc = b'd$ and $bd = b'c$. If the latter option is true, then $bcd = b'd^2 = b'c^2$, so $b' = 0$. In that case, $b = 0$, so $b = b'$. If the former option is true, because at least one of $c$ or $d$ is non-zero, again $b = b'$.

Let us now perform a coordinate transformation $z = Mx$, where $M$ is chosen so that $Mb=e_1$. Such an $M$ exists since $b\neq 0$. Then, $\dot{z}=MAx+Mbu=MAM^{-1}z+e_1u$. If we define $\hat{A}=MAM^{-1}$, by our assumptions $\hat{A}$ is invertible and $\hat{A}$ has distinct eigenvalues. Additionally, it is trivially verified that left eigenvectors of $\hat{A}$ are $(M^{-1})^T\eta_i$. Since $M^{-1}e_1=b$, the assumption $b^T\eta_i\neq 0$ is equivalent to the first element of the left eigenvectors of $\hat{A}$ being non-zero. To simplify the notation, by a standard abuse we now let $(A,e_1)$, $(A',e_1)$ represent the system matrices after performing the above transformation. By the above discussion, we are then assuming that $A$ and $A'$ are invertible, have distinct eigenvalues, and that $\eta_{i1}\neq 0$ for all $i$.

Noting that the two systems produce the same reachable sets, by \eqref{theoremABU} it follows that $A^ke_1=A'^ke_1$ for all $k\in\mathbb{N}$. By the same logic as in the first paragraph of the proof, we see that since $c \neq -d$, then $A^kce_1 = A^{\prime k}ce_1$ and $A^kde_1 = A^{\prime k}de_1$ is satisfied for all $k \in \mathbb{N}$, giving us the relation 

\begin{equation}\label{Lemma2_eq1}
    A^ke_1 = A^{\prime k} e_1\quad\forall\,k\in\mathbb{Z}_{\geq 0}.
\end{equation}

\noindent Equation \eqref{Lemma2_eq1} implies $A^{k-1}A^\prime e_1 = A^{\prime k-1}Ae_1$ and $A^{\prime k-2}Ae_1 = A^{k-1}e_1$ for all $k \geq 2$. We have $$A^{k-1}A^\prime e_1 = A^{\prime k-1}Ae_1 = A^\prime A^{\prime k-2}Ae_1 = A^\prime A^{k-1}e_1.$$ Hence, $A^{k-1}A^\prime e_1 = A^\prime A^{k-1}e_1$; since $A^\prime$ is invertible,

\begin{equation}\label{Lemma2_eq2}
    A^ke_1 = A^{\prime(-1)}A^kA^{\prime}e_1\quad\forall\,k\in\mathbb{Z}_{\geq 0}.
\end{equation}

Let $v_i$ denote the right eigenvectors of $A$ and $v_i',\,\eta_i'$ denote the right and left eigenvectors of $A^{\prime (-1)}AA^\prime$ respectively. Since $A$ and $A^{\prime (-1)}AA^\prime$ are similar matrices, their eigenvalues are equal \cite{strang2006linear}. Let $A = VDV^{-1}$ and $A'^{(-1)}AA' = V'DV'^{(-1)}$ where the rows of $V^{-1}$ and $V'^{-1}$ are $\eta_i^T$ and $\eta_i'^T$ respectively and the columns of $V$ and $V'$ are $v_i$ and $v_i'$ respectively. By our assumptions, $\eta_{i1} \neq 0$, so we can now scale the $\eta$'s so that $\eta_{i1} = 1$. We then redefine $v_i$ to be the newly scaled right eigenvectors such that $\eta_{i1} = 1$. Next, we write \eqref{Lemma2_eq2} in tensor notation \cite{strang2006linear} and get $$\sum_{i} \lambda_i^kv_i = \sum_{i} \lambda_i^kv_i^\prime\eta_{i1}'^T\quad\forall\,k\in\mathbb{Z}_{\geq 0}$$ which implies 

\begin{equation}\label{Thm2_eq1}
    \sum_i\lambda_i^k(v_i - v_i^\prime\eta_{i1}'^T) = 0\quad\forall\,k\in\mathbb{Z}_{\geq 0}.
\end{equation}

Taking $k \in \{0,\hdots,n-1\}$ we have a series of $n$ equations. For the $j$-th element of any $v_i$ and $v_i^\prime$, we have $$\Lambda S_j = \begin{bmatrix}1 & \hdots & 1\\ \lambda_1 & \hdots & \lambda_n \\ \vdots & \vdots & \vdots\\ \lambda_1^{n-1} & \hdots & \lambda_n^{n-1}\end{bmatrix}\begin{bmatrix}v_{1j} - v_{1j}^\prime\eta_{11}'^T \\ v_{2j} - v_{2j}^\prime\eta_{21}'^T\\ \vdots \\ v_{nj} - v_{nj}^\prime\eta_{n1}'^T\end{bmatrix} = \begin{bmatrix}0\\0\\\vdots\\0\end{bmatrix}$$ for any $j \in \{1,\hdots,n\}$. Notice that $\Lambda \in \mathbb{C}^{n\times n}$ is the square Vandermonde matrix \cite{klinger1967vandermonde}. Recall that the Vandermonde matrix is invertible if elements $\lambda_i$ are distinct for all $i$, which holds by assumption. If $\eta_{i1}' = 0$ for any $i$, then $v_i = 0$, which contradicts the assumption that $A$ is diagonalizable. Consequently, $\eta_{i1}' \neq 0$ for all $i$, so similar to the previous step, we can scale $v_i'$ and $\eta_{i1}'$ such that $\eta_{i1}' = 1$ for all $i$. It follows that $v_{ij} = v_{ij}^\prime$ for all $i,j$ since $\Lambda$ is invertible. Therefore, $A = A^{\prime(-1)}AA^\prime$. 

Recall that we assumed that all eigenvalues of $A$ are distinct. Thus, since $A$ and $A^\prime$ commute, we can conclude that $A$ and $A^\prime$ have the same eigenvectors \cite{strang2016introduction}. Recall that $A$ and $A^\prime$ are both diagonalizable. If we take the eigenvalue expansion of $A$ and $A'$ and multiply both on the left by $V^{-1}$, then equation \eqref{Lemma2_eq1} implies $$D^k\begin{bmatrix}\eta_{11}\\\eta_{21} \\ \vdots \\ \eta_{n1}\end{bmatrix}= D^{\prime k}\begin{bmatrix}\eta_{11}\\\eta_{21}\\ \vdots \\ \eta_{n1}\end{bmatrix}\quad \forall~k\in\mathbb{N},$$
where $D'$ is the diagonal matrix with eigenvalues of $A'$ on the diagonal. Subtracting the right hand side from both sides reveals that \eqref{Lemma2_eq1} implies $$(\lambda_i^k - \lambda_i^{\prime k})\eta_{i1} = 0\quad \forall~k \in \mathbb{N}.$$ By assumption, $\eta_{i1} \neq 0$ for all $i$, so $$\lambda_i^k = \lambda_i^{\prime k}\quad \forall~k \in \mathbb{N}.$$ Therefore, both $A$ and $A^\prime$ have the same eigenvectors and eigenvalues, hence $A = A^\prime$.
\end{proof}

Theorem \ref{Theorem_generalUniqueMatrix} proves that given reachable sets of generic system \eqref{discreteLinearSystem}, the pair $(A,b)$, i.e., the system dynamics, are uniquely defined when the set of control inputs is not symmetric around $0$. We now want to address the degenerate case where $\mathcal{U} = [-c,c]$. It can be easily seen that in such a case, system \eqref{discreteLinearSystem} with $(A,b)$ and $(-A,-b)$ will produce the same reachable sets. To discuss a relaxed notion of system uniqueness, we provide a formal definition of $\pm$-uniqueness. 

\begin{mydef}
The system dynamics $(A,b)$ of \eqref{discreteLinearSystem} are \mbox{$\pm$-unique} if $(A,b)$ and $-(A,b)$ generate the same reachable sets, but there do not exist other pairs $(A^\prime,b^\prime)$ which generate the same reachable sets.
\end{mydef}

We conjecture that in the case when $\mathcal{U}$ is symmetric around $0$ -- a scenario common in many controls applications \cite{chen1984linear} -- the dynamics are $\pm$-unique.

\begin{conj}\label{Conjecture}
Let $\mathcal{U} = [-c,c]$. Let the sequence $\{\mathcal{R}(i,0)\}_{i\in\mathbb{N}}$ be generated by $(A,b)$, where $A^2$ has distinct eigenvalues and $(A,b)$ are known to satisfy the assumptions of Theorem \ref{Theorem_generalUniqueMatrix}. Then, $(A,b)$ is $\pm$-unique.
\end{conj}

Proving the conjecture above requires extensive theoretical developments and remains for future work. As an illustration, we formally prove the conjecture to be true in the two-dimensional case. 

\begin{thm}\label{Theorem_genericUniqueMatrix_Symmetric}
Let $n=2$. Then, Conjecture \ref{Conjecture} is correct.
\end{thm}
\begin{proof}
Similarly to the proof of Theorem \ref{Theorem_generalUniqueMatrix}, we have two options: $bc = b'c$ or $bc = -b'c$. In the former case, we reach the same result as before, namely $b = b'$. In the latter case, we obtain $b = -b'$. Altogether, we get $b = (-1)^{p(0)}b'$ where $p(0) \in \{0,1\}$.

As in Theorem \ref{Theorem_generalUniqueMatrix}, through a coordinate transformation, we assume without loss of generality that $b' = e_1$. Then $b\mathcal{U} = (-1)^{p(0)}b^\prime\mathcal{U} = (-1)^{p(0)}[-c,c]e_1$. Following the same steps as in the beginning of the proof in Theorem \ref{Theorem_generalUniqueMatrix}, with a standard abuse of notation, we let $A$, $A^\prime$ represent the system dynamics in this new basis where $A$ and $A'$ satisfy our assumptions. Also, we find that if $\mathcal{U} = [-c,c]$, then we arrive at the relation 

\begin{equation}\label{Theorem3_eq1}
    A^ke_1 = (-1)^{p(k)}A^{\prime k} e_1\quad \forall\,k\in\mathbb{Z}_{\geq 0}.
\end{equation}

When $k = 2$, we see that regardless of $p(1)$, $AA'e_1 = (-1)^{p(2)}A'Ae_1$. Using this fact along with equation \eqref{Theorem3_eq1} implies $A^{k-1}A^\prime e_1 = (-1)^{p(k)}A^{\prime k-1}Ae_1$ for all $k \geq 1$ and $A^{k-2}A^\prime e_1 = (-1)^{p(k-1)}A^{\prime k-2}Ae_1$ for all $k \geq 2$. We have $$(-1)^{p(k)}A^{\prime k-1}Ae_1 = (-1)^{p(k)}A^\prime A^{\prime k-2}Ae_1$$ $$=(-1)^{p(k)}(-1)^{p(k-1)}(-1)^{p(1)}A^\prime A^{k-1}e_1.$$ Hence, $A^{k-1}A^\prime e_1 = (-1)^{p(k)}(-1)^{p(k-1)}(-1)^{p(1)}A^\prime A^{k-1}e_1$; since $A^\prime$ is invertible, $$A^{k-1}e_1 = \frac{A'^{(-1)} A^{k-1}A^\prime e_1}{(-1)^{p(k)}(-1)^{p(k-1)}(-1)^{p(1)}}.$$ We define $q(k) \in \{0,1\}$ by $(-1)^{q(k)}~=~((-1)^{p(k)}(-1)^{p(k-1)}(-1)^{p(1)})^{-1} = (-1)^{p(k)}(-1)^{p(k-1)}(-1)^{p(1)}$. We then have $$A^{k-1}e_1 = (-1)^{q(k)}A^{\prime(-1)}A^{k-1}A^{\prime}e_1\quad\forall\,k\in\mathbb{N}.$$ It holds that $A^{k-1} = A'^{(-1)}A^{k-1}A'$ have the same eigenvalues, so $A^{k-1} = -A'^{(-1)}A^{k-1}A'$ must have eigenvalues of opposite sign. That is, if $\lambda_i$ and $\lambda_i'$ are the eigenvalues of $A$ and $\pm A'^{(-1)}AA'$ respectively, then $\lambda_i = \pm \lambda_i'$. Following the same steps as in the proof of Theorem \ref{Theorem_generalUniqueMatrix} we get $$\sum_{i} \lambda_i^{k-1}v_i = (-1)^{q(k)} \sum_{i} \lambda_i^{k-1}v_i^\prime\eta_{i1}'^T\quad\forall\,k\in\mathbb{N}.$$ Subtracting the right hand side from both sides gives us

\begin{equation}\label{Theorem3_eq2}
    \sum_i\lambda_i^{k-1}(v_i - (-1)^{q(k)} v_i^\prime\eta_{i1}'^T) = 0\quad\forall\,k\in\mathbb{N}.
\end{equation}

We now show that if $(A,B)\in(\mathbb{R}^{2\times 2},\mathbb{R}^2)$, then equation \eqref{Theorem3_eq2} implies $A = \pm A^{\prime(-1)}AA^\prime$. Recall that $q(k) \in \{0,1\}$ and so $(q(1),q(2)) \in \{(0,0),\,(0,1),\,(1,0),\,(1,1)\}$. When $(q(1),q(2)) = (0,0)$, equation \eqref{Theorem3_eq2} is the same as equation \eqref{Thm2_eq1} for $k=1$ and $k=2$. If we write these equations in matrix form as in Theorem \ref{Theorem_generalUniqueMatrix} we again have the Vandermonde matrix on the left-hand side. Following the same steps as Theorem \ref{Theorem_generalUniqueMatrix}, we see that $v_i = v_i'$ for all $i$. Since $\lambda_i = \pm \lambda_i'$, then $A = \pm A^{\prime(-1)}AA^\prime$. Similarly, if $(q(1),q(2)) = (1,1)$, if we follow the same procedure to find $A = \pm A^{\prime(-1)}AA^\prime$.

The most interesting cases are when $(q(1),q(2)) \in \{(0,1),(1,0)\}$. Let us first consider $(q(1),q(2)) = (1,0)$. Recall $q(k) \in \{0,1\}$, so if $q(3) = 0$, then $(q(2),q(3)) = (0,0)$. If $(q(k),q(k+1)) = (0,0)$ for some $k$, we then have

\begin{equation}\label{Thm2_eq2}
    \Lambda S_j = \begin{bmatrix}\lambda_1^k & \lambda_2^k \\ \lambda_1^{k+1} & \lambda_2^{k+1}\end{bmatrix} \begin{bmatrix}v_{1j} - v^\prime_{1j}\eta_{11}'^T\\v_{2j} - v^\prime_{2j}\eta_{21}'^T\end{bmatrix} = \begin{bmatrix}0\\0\end{bmatrix}.
\end{equation}

\noindent We note $$\mathrm{det}(\Lambda) = \lambda_1^k\lambda_2^k\begin{vmatrix}1 & 1 \\ \lambda_1 & \lambda_2\end{vmatrix} \neq 0$$ since we have two non-zero scalars multiplied by the non-zero Vandermonde determinant in the case of distinct eigenvalues. Hence, $\Lambda$ as defined in \eqref{Thm2_eq2} is invertible and we again conclude that $A = \pm A^{\prime(-1)}AA^\prime$. 

We lastly consider cases where $q(k)$ is alternating, namely $\{q(k)\}_{k=1}^3 = (0,1,0)$ and $\{q(k)\}_{k=1}^3 = (1,0,1)$. In the former case, we have $$\Lambda S_j = \begin{bmatrix}1 & 1\\ \lambda_1^2 & \lambda_2^2\end{bmatrix} \begin{bmatrix}v_{1j} - v^\prime_{1j}\eta_{11}'^T\\v_{2j} - v^\prime_{2j}\eta_{21}'^T\end{bmatrix} = \begin{bmatrix}0\\0\end{bmatrix}.$$ The generic assumption that all eigenvalues are distinct modulo a sign implies $\Lambda$ is invertible, thus we again find $v_i = v_i'$ and thus $A = \pm A^{\prime(-1)}AA^\prime$. By following the same steps, we arrive at the same conclusion when $\{q(k)\}_{k=1}^3 = (1,0,1)$.

We now have that $A = (-1)^{q(2)} A^{\prime(-1)}AA^\prime$. If $q(2) = 0$, then $A$ and $A^\prime$ commute. Using assumptions of the theorem statement, we can conclude that $A$ and $A^\prime$ have the same eigenvectors \cite{strang2016introduction}. If $q(2) = 1$, then $A = -A^{\prime(-1)}AA^\prime$ and so $A^2 = A^{\prime(-1)}A^2A^\prime$. Clearly, $A^2$ and $A^\prime$ commute, and by the theorem statement, $A^2$ has distinct eigenvalues, which again implies that $A$ and $A^\prime$ share the same eigenvectors.

We now follow the same steps as in the latter part of the proof of Theorem \ref{Theorem_generalUniqueMatrix}. Namely, we can diagonalize $A$ and $A'$; taking the eigenvalue expansion of equation \eqref{Theorem3_eq1} and multiplying both sides on the left by the matrix of left eigenvectors gives us the series of equations $$\lambda_i^{k-1} = (-1)^{q(k)}\lambda_i^{\prime k-1}\quad \forall~k \in \mathbb{N}.$$ Since $q(2) = 0$ or $q(2) = 1$, then $\lambda_i = \lambda_i^\prime$ or $\lambda_i = -\lambda_i^\prime$ for all $i$. Since both $A$ and $A^\prime$ have the same eigenvectors and eigenvalues same to a sign, then $A = \pm A^\prime$.
\end{proof}

Theorem \ref{Theorem_generalUniqueMatrix} solves Problem \ref{problemStatement} in the generic case where $\mathcal{U} \neq [-c,c]$ while Theorem \ref{Theorem_genericUniqueMatrix_Symmetric} proves there exists a $\pm$-unique solution to Problem \ref{problemStatement} in the two-dimensional case where $\mathcal{U} = [-c,c]$. The proof of Theorem \ref{Theorem_genericUniqueMatrix_Symmetric} drives our intuition for Conjecture \ref{Conjecture} in general: intuitively, adding dimensions to the system should not make it more likely that multiple generic systems can produce the same reachable sets for all time, especially considering no two such systems exist when the input set is asymmetric. Formalizing this statement is left for future work. 

We remark that if the system dynamics do not satisfy the assumptions of Theorem \ref{Theorem_generalUniqueMatrix} or Theorem \ref{Theorem_genericUniqueMatrix_Symmetric}, they might not be (uniquely or $\pm$-uniquely) recoverable. However, using a slight perturbation of the reachable sets might recover a generic approximation of the true dynamics. Doing so, however, introduces challenges on the method of perturbing these sets. We leave such a discussion for future work.
\section{Solving for the System Dynamics}

We ultimately want to use reachable sets to solve for the system dynamics. Equation \eqref{theoremABU} of Theorem \ref{discreteMinkowskiDifference} already gives us a formula for calculating $A^{i-1}b\mathcal{U}$ for all $i \in \mathbb{N}$, namely $$A^{i-1}b\mathcal{U} = \mathcal{R}(i,0) \ominus \mathcal{R}(i-1,0).$$ In Theorem \ref{Theorem_generalUniqueMatrix}, we proved that the answer to Problem \ref{problemStatement} is affirmative for generic, single-input linear systems, meaning that for cases where the linear system dynamics satisfy the generic assumptions of Lemma \ref{matrixProperties_GeneralCase}, we can uniquely determine the true dynamics from the system's reachable sets. This motivates us to devise a procedure to calculate $(A,b)$. 

We will determine $(A,b)$ from reachable sets through a two step procedure. First, we calculate $A^{i-1}b\mathcal{U}$ for $i~=~\{1,\hdots,n + 1\}$. In the case where $\mathcal{U} \neq [-c,c]$, the sequence of sets $A^{i-1}b\mathcal{U}$ can be used to calculate $(A,b)$ directly. If $\mathcal{U} = [-c,c]$, these same sets can be utilized to compute a number of candidate dynamics $(A,b)$ which satisfy $\mathcal{R}(i,0)$ for all $i$. To determine which candidate solutions are correct, we compute the forward reachable sets of \eqref{discreteLinearSystem} using all candidate $(A,b)$. By Theorem \ref{Theorem_genericUniqueMatrix_Symmetric}, in the two-dimensional case, only two solutions $(A,b)$ and $(A^\prime,b^\prime)$ such that $(A,b) = -(A^\prime,b^\prime)$ will satisfy $\mathcal{R}(i,0)$ for all $i$.

We begin our method by first using an algorithm that takes reachable sets of \eqref{discreteLinearSystem} and solves for $A^{i-1}b\mathcal{U}$. By equation \eqref{theoremABU}, we can utilize existing methods \cite{althoff2015computing, yan2015closed, feng2019minkowski} to compute the Minkowski difference between two polygons to calculate $A^{i-1}b\mathcal{U}$ given $\mathcal{R}(i,0)$ for all $i \in \mathbb{N}$. For this narrative, we adopt the method in \cite{althoff2015computing}. By Lemma 1 of \cite{althoff2015computing}, if we let $v^{(i)} \in \mathcal{V}$ be the vertices of $\mathcal{R}(i-1,0)$, then the Minkowski difference $\mathcal{R}(i,0) \ominus \mathcal{R}(i-1,0)$ may be computed by taking the intersection of the translation of the set $\mathcal{R}(i,0)$ by vertices $v^{(i)} \in \mathcal{V}$ of $\mathcal{R}(i-1,0)$: 

\begin{equation}\label{minkowskiDifference_Vertices}
    \mathcal{R}(i,0) \ominus \mathcal{R}(i-1,0) = \bigcap_{v^{(i)} \in \mathcal{V}}(\mathcal{R}(i,0) - v^{(i)}).
\end{equation}

\noindent While computing the intersection in \eqref{minkowskiDifference_Vertices} is generally computationally difficult, calculations are made significantly easier as $A^{i-1}b\mathcal{U}$ is a line segment; see \cite{althoff2015computing} for details.

We now move to recover $A^{i-1}b$ from $A^{i-1}b\mathcal{U}$. We consider two cases: $\mathcal{U} \neq [-c,c]$ and $\mathcal{U} = [-c,c]$ for some $c \in \mathbb{R}$. In the former case, taking the mean of the vertices of $A^{i-1}b\mathcal{U}$ will provide $A^{i-1}b\frac{c + d}{2}$. Multiplying this vector by $\frac{2}{c+d}$ recovers $A^{i-1}b$.

\begin{thm}\label{calculateA}
Let us assume the $n$-dimensional system \eqref{discreteLinearSystem} is controllable. Let $C_{A,b} = \begin{bmatrix}b & Ab & \hdots & A^{n-1}b\end{bmatrix}$. For the single-input case, $A = AC_{A,b}C_{A,b}^{-1}$. 
\end{thm}

The proof of Theorem \ref{calculateA} is trivial, noting that $C_{A,b}$ is full rank for controllable systems. We note that the assumption of controllability is generic \cite{chen1984linear}.

In the case where $\mathcal{U} = [-c,c]$, by multiplying the vertices of $A^{i-1}b\mathcal{U}$ by $c$, we can only recover $A^{i-1}b$ up to a sign, generating two candidates for each $i$. Substituting all possible candidates for $A^{i-1}b$ into the columns of $C_{A,b}$ and $AC_{A,b}$ generates $2^{n+1}$ candidate matrices $A$.

To determine which candidate solutions yield the correct $\pm$-unique matrix pair $(A,b)$, we can plot the reachable sets of all $2^{n+1}$ candidate solutions to solve for the desired unknown $\pm$-unique system dynamics. In the next section, we use the CORA toolkit \cite{althoff2018implementation} and adopt methods of computing the Minkowski difference detailed in \cite{althoff2015computing} to numerically calculate the dynamics $(A,b)$ for an unknown band-pass filter circuit system and a two-dimensional unknown system with $\mathcal{U} = [-1,1]$, validating the developed theory.

\section{Numerical Examples}


To validate the developed theory and demonstrate how to apply the proposed method, we first consider a scenario of reverse engineering an electric circuit from manufacturer specifications. At times, manufacturers will only release partial information about a system. For example, instead of providing a dynamic model of a manufactured part, manufacturers might convey the set of all voltages a circuit may output within a set amount of time given the set of all viable input frequencies. Such information can be interpreted as the minimum time in which a state can be reached, providing a picture of the system's reachable sets. Motivated by such an example, in this section, we provide an example of identifying the matrices $(A,b)$ of a band-pass filter circuit from its reachable sets. In a subsequent example, we identify the $\pm$-unique dynamics of an unknown two-dimensional system with an input set symmetric around zero. Both examples utilize the CORA toolkit \cite{althoff2018implementation} for set computations, namely to calculate convex hulls and Minkowski differences.
\vskip -10pt

\subsection{Band-Pass Filter Circuit}

We present the linear dynamic model of a band-pass filter circuit \cite{denisenko2021synthesis}. Let us assume $x[0] = 0$. The state-space controllable canonical representation \cite{chen1984linear} of this circuit is

\vskip -10pt

\begin{equation}\label{Band-Pass Filter State Space}
\begin{gathered}
    x[j+1] = Ax[j] + bv_c[j] \\ = \begin{bmatrix}0 & 1 & 0 & 0\\ 0 & 0 & 1 & 0\\ 0 & 0 & 0 & 1\\ -a_0 & -a_1 & -a_2 & -a_3\end{bmatrix}x[j] + \begin{bmatrix}0\\0\\0\\1\end{bmatrix}v_c[j]
\end{gathered}
\end{equation}

\noindent such that $v_c[j] \in [0,1]$ for all $j \in \mathbb{Z}_{\geq 0}$. 

Assume the reachable sets $\{\mathcal{R}(j+1,0)\}_{j=0}^\infty$ of the controllable, dynamical system \eqref{Band-Pass Filter State Space} are known. From controllability, we know the form of system \eqref{Band-Pass Filter State Space}, but not parameters $a_0,\,a_1,\,a_2,\,a_3$. From this information, we want to recover the true parameters: $a_0 = 3$, $a_1 = 2$, $a_2 = 3$, and $a_3 = 6$. It can be easily shown that if $a_0 \neq 0$, the assumptions of Theorem \ref{Theorem_generalUniqueMatrix} are satisfied. Clearly, $A$ from \eqref{Band-Pass Filter State Space} satisfies said assumptions. Clearly, the matrix $M$ for which $Mb = e_1$ is a simple row permutation; the assumptions of Theorem \ref{Theorem_generalUniqueMatrix} are invariant under permutations, hence all assumptions are satisfied and the results of Theorem \ref{Theorem_generalUniqueMatrix} apply when solving for the matrix pair $(A,b)$. That is, there exists a unique matrix pair which satisfies $\{\mathcal{R}(j+1,0)\}_{j=0}^\infty$. Since the system is four-dimensional, Theorem \ref{calculateA} shows we need only consider the sets $\{\mathcal{R}(j+1,0)\}_{j=0}^4$ to calculate $(A,b)$.

Assume, that $\{\mathcal{R}(j+1,0)\}_{j=0}^4$ are known to equal \small $$\mathcal{R}(1,0) = \mathrm{conv}\left(\begin{bmatrix}0\\0\\0\\0\end{bmatrix},\begin{bmatrix}0\\0\\0\\1\end{bmatrix}\right),$$ $$\mathcal{R}(2,0) = \mathrm{conv}\left(\begin{bmatrix}0\\0\\1\\-5\end{bmatrix},\begin{bmatrix}0\\0\\1\\-6\end{bmatrix},\begin{bmatrix}0\\0\\0\\0\end{bmatrix},\begin{bmatrix}0\\0\\0\\1\end{bmatrix}\right),$$ $$\mathcal{R}(3,0) = \mathrm{conv}\left(\begin{bmatrix}0\\0.86\\-6.02\\33.00\end{bmatrix},\begin{bmatrix}0\\-0.14\\0.98\\-6.00\end{bmatrix},\begin{bmatrix}0\\-0.14\\0.98\\-5.00\end{bmatrix},\begin{bmatrix}0\\0.86\\-6.02\\34.00\end{bmatrix}\right),$$
$$\mathcal{R}(4,0) = \mathrm{conv}\left(\begin{bmatrix}-0.15\\1.05\\-5.99\\33.00\end{bmatrix},\begin{bmatrix}0.85\\-5.95\\34.01\\-188\end{bmatrix},\begin{bmatrix}0.85\\-5.95\\34.01\\-187\end{bmatrix},\begin{bmatrix}-0.15\\1.05\\-5.99\\34.00\end{bmatrix}\right),$$ 
$$\mathcal{R}(5,0) =$$ \vskip -20pt $$\mathrm{conv}\left(\begin{bmatrix}-5.93\\33.88\\-188.02\\1035.00\end{bmatrix},\begin{bmatrix}1.07\\-6.12\\33.98\\-188.00\end{bmatrix},\begin{bmatrix}1.07\\-6.12\\33.98\\-187.00\end{bmatrix},\begin{bmatrix}-5.93\\33.88\\-188.02\\1036.00\end{bmatrix}\right).$$\normalsize Based on Theorem \ref{Theorem_generalUniqueMatrix}, the knowledge of only these five sets is sufficient to reconstruct the true values of parameters $a_0$, $a_1$, $a_2$, and $a_3$.

Since $\mathcal{R}(1,0) = b\mathcal{U}$ and $\mathcal{U} = [0,1]$, $b$ can be trivially computed to equal $b = \begin{bmatrix}0 & 0 & 0 & 1\end{bmatrix}^T$. Next, by equation \eqref{theoremABU} of Theorem \ref{discreteMinkowskiDifference}, $Ab\mathcal{U} = \mathcal{R}(2,0) \ominus \mathcal{R}(1,0)$. Given $\mathcal{U} = [0,1]$, by taking the Minkowski difference we get $Ab = \begin{bmatrix}0 & 0 & 1 & -6\end{bmatrix}^T$. Repeating the procedure, we have $A^2b\mathcal{U} = \mathcal{R}(3,0) \ominus \mathcal{R}(2,0)$, $A^3b\mathcal{U} = \mathcal{R}(4,0) \ominus \mathcal{R}(3,0)$, $A^4b\mathcal{U} = \mathcal{R}(5,0) \ominus \mathcal{R}(4,0)$, and $\mathcal{U} = [0,1]$. It follows that $$A^2b = \begin{bmatrix}0\\1\\-6\\33\end{bmatrix},\,A^3b = \begin{bmatrix}1\\-6\\33\\-182\end{bmatrix},\,A^4b = \begin{bmatrix}-6\\33\\-182\\1002\end{bmatrix}.$$ Recall we assume the system is controllable, and thus the controllability matrix $C_{A,b}$ is invertible. Finally, by \mbox{Theorem \ref{calculateA}}, $$A = AC_{A,b}C_{A,b}^{-1}$$ $$= \begin{bmatrix}A^4b & A^3b & A^2b & Ab\end{bmatrix}\begin{bmatrix}A^3b & A^2b & Ab & b\end{bmatrix}^{-1}$$ which produces the correct matrix $A$ accurately reconstructing the parameters $a_0 = 3$, $a_1 = 2$, $a_2 = 3$, and $a_3 = 6$.

\subsection{Numerical Example with a Symmetric Input Set}

To validate Theorem \ref{Theorem_genericUniqueMatrix_Symmetric}, we present an example of a linear two-dimensional dynamical system

\begin{equation}\label{academicSystem Dynamics}
    \begin{bmatrix}x_1[i + 1] \\ x_2[i + 1]\end{bmatrix} = \begin{bmatrix}2 & 1\\ 2 & 3\end{bmatrix}\begin{bmatrix}x_1[i]\\x_2[i]\end{bmatrix} + \begin{bmatrix}0 \\ 1\end{bmatrix}u[i]
\end{equation}

\noindent with $\mathcal{U} = [-1,1]$. Such a system satisfies the assumptions of Lemma \ref{matrixProperties_GeneralCase}. As in the previous example, we will show that we can reconstruct the values of system matrices in \eqref{academicSystem Dynamics} from reachable sets, albeit up to a sign. Assume, thus, that we are given a sequence of reachable sets $\{\mathcal{R}(i,0)\}_{i=1}^4$ as convex hulls of vertices:

$$\mathcal{R}(1,0) = \mathrm{conv}\left(\pm\begin{bmatrix}0\\1\end{bmatrix}\right),$$
$$\mathcal{R}(2,0) = \mathrm{conv}\left(\pm\begin{bmatrix}-1\\-4\end{bmatrix},\pm\begin{bmatrix}1\\2\end{bmatrix}\right),$$
$$\mathcal{R}(3,0) = \mathrm{conv}\left(\pm\begin{bmatrix}6\\15\end{bmatrix},\pm\begin{bmatrix}4\\7\end{bmatrix},\pm\begin{bmatrix}6\\13\end{bmatrix}\right),$$
$$\mathcal{R}(4,0) = \mathrm{conv}\left(\pm\begin{bmatrix}27\\58\end{bmatrix},\pm\begin{bmatrix}15\\28\end{bmatrix},\pm\begin{bmatrix}25\\50\end{bmatrix},\pm\begin{bmatrix}27\\56\end{bmatrix}\right).$$

\noindent Clearly, $b\mathcal{U} = \mathcal{R}(1,0)$. Equation \eqref{theoremABU} of Theorem \eqref{discreteMinkowskiDifference} also shows that $Ab\mathcal{U} = \mathcal{R}(2,0) \ominus \mathcal{R}(1,0)$, $A^2b\mathcal{U} = \mathcal{R}(3,0) \ominus \mathcal{R}(2,0)$, and $A^3b\mathcal{U} = \mathcal{R}(4,0) \ominus \mathcal{R}(3,0)$. Since $\mathcal{U} = [-1,1]$, through the same calculations in the previous example we thus obtain $$b = \pm\begin{bmatrix}0\\1\end{bmatrix},\, Ab = \pm\begin{bmatrix}1\\3\end{bmatrix},$$ $$A^2b = \pm\begin{bmatrix}5\\11\end{bmatrix},\, A^3b = \pm\begin{bmatrix}21\\43\end{bmatrix}.$$

Let us denote $b^- = \begin{bmatrix}0\\-1\end{bmatrix}$, $b^+ = \begin{bmatrix}0\\1\end{bmatrix}$, $Ab^- = \begin{bmatrix}-1\\-3\end{bmatrix}$, etc. We now consider a set of $2^3$ possible candidate pairs of $(C_{A,b},AC_{A,b})$ matrices:

\begin{equation}\label{candidateControllabilityMatrices}
    (C_{A,b},AC_{A,b}) = \left\{\begin{array}{cc}
        \left( \begin{bmatrix}b^+ & Ab^+\end{bmatrix}, \begin{bmatrix}Ab^+ & A^2b^+\end{bmatrix} \right),\\
        \left( \begin{bmatrix}b^- & Ab^+\end{bmatrix}, \begin{bmatrix}Ab^+ & A^2b^+\end{bmatrix} \right),\\
        \left( \begin{bmatrix}b^+ & Ab^-\end{bmatrix}, \begin{bmatrix}Ab^- & A^2b^+\end{bmatrix} \right),\\
        \left( \begin{bmatrix}b^- & Ab^-\end{bmatrix}, \begin{bmatrix}Ab^- & A^2b^+\end{bmatrix} \right), \\
        \left( \begin{bmatrix}b^+ & Ab^+\end{bmatrix}, \begin{bmatrix}Ab^+ & A^2b^-\end{bmatrix} \right),\\
        \left( \begin{bmatrix}b^- & Ab^+\end{bmatrix}, \begin{bmatrix}Ab^+ & A^2b^-\end{bmatrix} \right),\\
        \left( \begin{bmatrix}b^+ & Ab^-\end{bmatrix}, \begin{bmatrix}Ab^- & A^2b^-\end{bmatrix} \right),\\
        \left( \begin{bmatrix}b^- & Ab^-\end{bmatrix}, \begin{bmatrix}Ab^- & A^2b^-\end{bmatrix} \right).
    \end{array}\right.
\end{equation}

\noindent By Theorem \ref{calculateA}, determining all candidate matrix pairs $(A,b)$ becomes a trivial calculation using all possible pairs from \eqref{candidateControllabilityMatrices}. Doing so provides two $\pm$-unique candidate pairs: \small$$(A,b) = \pm\left(\begin{bmatrix}2 & 1\\ 2 & 3\end{bmatrix},\begin{bmatrix}0\\1\end{bmatrix}\right),\,(A',b') = \pm\left(\begin{bmatrix}8 & -1\\ 20 & -3\end{bmatrix},\begin{bmatrix}0\\-1\end{bmatrix}\right).$$\normalsize

\begin{figure}[h]
	\centering\vspace{0.1cm}
	\subfloat[$\mathcal{R}(4,0)$ for $(A,b)$\centering]{{\includegraphics[width=4.35cm]{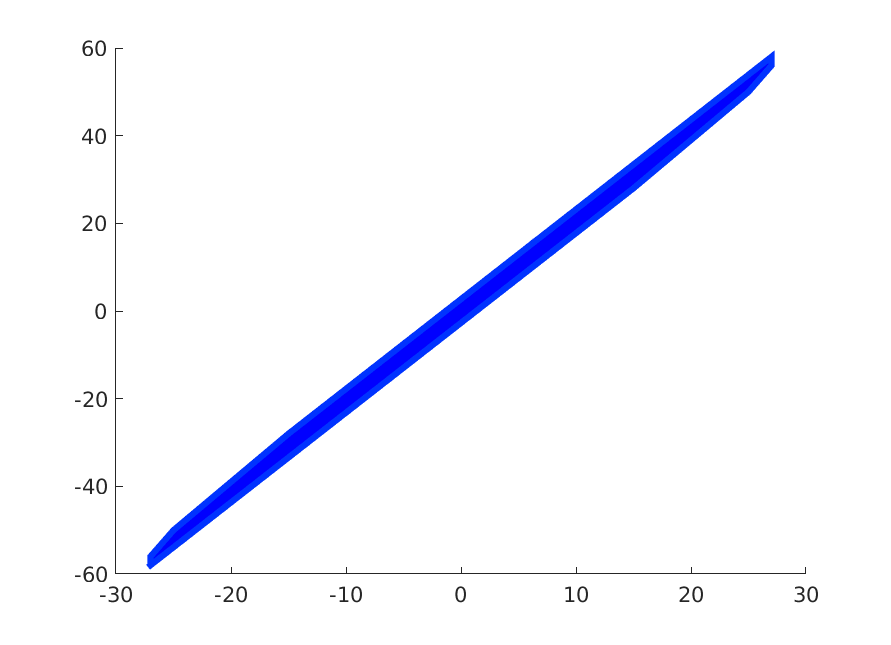} }}%
    \subfloat[$\mathcal{R}'(4,0)$ for $(A',b')$\centering]{{\includegraphics[width=4.35cm]{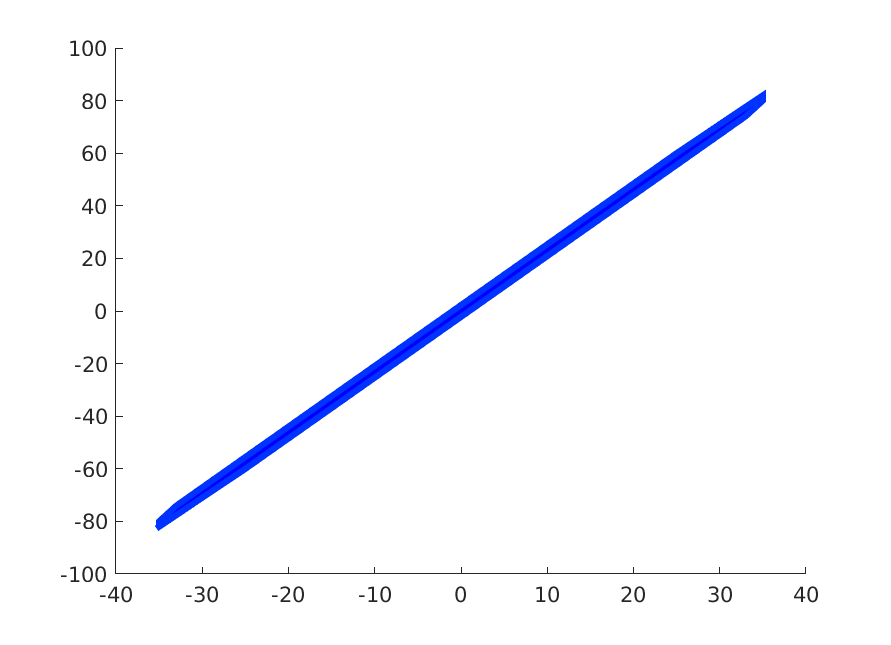} }}%
	\caption{Reachable sets for candidate dynamics at $i = 4$.}
    \label{fig:CandidateReachableSets}
\end{figure}

While the calculations above used only $\mathcal{R}(1,0)$, $\mathcal{R}(2,0)$, and $\mathcal{R}(3,0)$, to distinguish between these two final candidates we need to employ $\mathcal{R}(4,0)$. Fig. \ref{fig:CandidateReachableSets} shows the plots of the forward reachable sets for system \eqref{academicSystem Dynamics} at time $i = 4$ with matrix pairs $(A,b)$ and $(A',b')$ on the left and right respectively.

Fig. \ref{fig:CandidateReachableSets}(a) shows a reachable set that is identical to $\mathcal{R}(4,0)$, while Fig. \ref{fig:CandidateReachableSets}(b) illustrates the reachable set $$\mathcal{R}'(4,0) = \mathrm{conv}\left(\pm\begin{bmatrix}35\\82\end{bmatrix},\,\pm\begin{bmatrix}25\\60\end{bmatrix},\,\pm\begin{bmatrix}33\\74\end{bmatrix},\,\pm\begin{bmatrix}35\\80\end{bmatrix}\right),$$ which is not the same as $\mathcal{R}(4,0)$. Therefore, we can identify the matrix pair $(A,b)$, up to a sign, as the true dynamics of the unknown linear system \eqref{academicSystem Dynamics}. As mentioned before, reachable sets in this case do not allow us to distinguish any further: the dynamics that differ only in the sign generate the same reachable sets.
\section{Conclusion}

This paper considers the problem of determining the dynamics of an unknown discrete-time linear system using its reachable sets. The theory developed in this paper proves that for input sets that are asymmetric around the origin, the derived system dynamics are, given some technical assumptions, unique. Thus, in such cases, we can determine the true dynamics of an unknown system using the sequence of the system's reachable sets. For the case where the input set is symmetric, we prove that the derived dynamics are unique up to a factor of $\pm 1$ for two-dimensional systems and provide a conjecture that asserts the same result holds for the $n$-dimensional case. We then develop a method for deriving the dynamics of a system given the sequence of the system's reachable sets using Minkowski differences and proceed to illustrate by example how the method can be applied to identify the unknown linear model of a band-pass filter. Novel identification methods are also applied to an academic system with an input set symmetric around zero to detail how we can adapt said methods to uniquely identify the model of a linear system. 

A natural next step is to prove the stated conjecture to show $\pm$-uniqueness for $n$-dimensional systems. Also, our current technical assumptions are consistent with generic properties of matrices, but ideally we want to relax these assumptions to identify necessary conditions for uniqueness. We also want to consider cases when the state's initial conditions are non-zero, when there is only available knowledge of the system's reachable sets at non-consecutive time steps, and also when working with the more general framework of multi-input systems. 
\section*{Acknowledgements}

We thank Jeffrey Stuart from Pacific Lutheran University for providing insights in combinatorial matrix theory that helped develop the scope this project, namely addressing the question of uniqueness outlined in Problem \ref{problemStatement}. 

\bibliographystyle{IEEEtran}
\bibliography{root.bib}

\end{document}